\documentclass[letterpaper, 10 pt, conference]{ieeeconf}  

\IEEEoverridecommandlockouts                              

\usepackage{ntheorem}
\usepackage{subfigure, amsmath, rotating}
\usepackage{array, multirow}
\usepackage{amsbsy, amsfonts, graphics, graphicx}
\usepackage{algorithmic}
\usepackage{algorithm}
\usepackage{epstopdf}
\usepackage{mathrsfs}
\usepackage{graphicx}
\usepackage{caption}
\usepackage{amssymb}
\usepackage{amsmath}

\usepackage[multiple]{footmisc}
\usepackage{mathtools}
\DeclarePairedDelimiter\ceil{\lceil}{\rceil}
\DeclarePairedDelimiter\floor{\lfloor}{\rfloor}
\usepackage{tikz}
\usepackage{float}
\usetikzlibrary{calc,patterns,decorations.pathmorphing,decorations.markings}

\newtheorem{assumption}{Assumption}
\newtheorem{theorem}{Theorem}

\newtheorem{remark}{Remark}

\newtheorem{corollary}{Corollary}

\title{\LARGE \bf Real-time Distributed MPC for Multiple Underwater Vehicles with Limited Communication Data-rates}

\author{Yujia Yang, Ye Wang, Chris Manzie, and Ye Pu
\thanks{$^{1}$Y. Yang, Y. Wang, C. Manzie, and Y. Pu are with the Department of Electrical and Electronic Engineering, University of Melbourne, Parkville VIC 3010, Australia {\tt\small {yujyang1}@student.unimelb.edu.au, {ye.wang1,manziec,ye.pu}@unimelb.edu.au}}
\thanks{Y. Yang is supported by the Melbourne Research Scholarship provided by the University of Melbourne.}
}
\begin{document}

\maketitle
\thispagestyle{empty}
\pagestyle{empty}

\begin{abstract}

Controlling a fleet of autonomous underwater vehicles can be challenging due to low bandwidth communication between agents. This paper proposes to address this challenge by optimizing the quantization design of the communications between agents for use in distributed algorithms. The proposed approach considers two stages for the problem of multi-AUV control: an off-line stage where the quantization design is optimized; and an on-line stage based on a distributed model predictive control formulation and a distributed optimization algorithm with quantization. The standard properties of recursive feasibility and stability of the closed loop systems are analyzed, and simulations used to demonstrate the overall behaviour of the proposed approach.

\end{abstract}

\section{Introduction}
Cooperative control of multiple autonomous underwater vehicles (AUVs) has been drawing increasing attention in ocean exploration, where sub-meter resolution scan of vast range of sea floor is desired \cite{8604772}, \cite{Rego2014}. The successful implementation of a cooperative control strategy depends on reliable communications. However, marine environment imposes strict constraints on the communication range and data-rate \cite{Fallon2010}. In practice, acoustic modems only have communication range between 0.1 km to 5 km and data-rate between 0.1 kbps to 15 kbps \cite{Sendra2016}. Consequently, these communication constraints lead to inaccurate information exchange that may cause deteriorated control performance \cite{Wen2019}. 

Model Predictive Control (MPC) is a popular tool for motion control of autonomous systems. It can optimize multiple control specifications taking into account, for example, state and input constraints, collision avoidance and energy consumption \cite{Yang2019}. For controlling a multi-AUV system with a surfacing unit, a centralized MPC can be used by solving the resulted optimization problem in the surfacing unit who gathers global information \cite{Fallon2009}. However, for a multi-AUV system operating in deep water, long range information exchange is less reliable, sometimes impossible. As the system gets bigger and more complex, solving an optimal control problem in a centralized way becomes difficult, since it requires full communication to collect information from each subsystem, and enough computational power on one central entity to solve the global optimization problem. A promising concept to avoid these problems is to use distributed MPC (DMPC) technique, requiring only neighbour-to-neighbour communication, to solve network-level control problems. The recent results have shown the great potential of DMPC for motion control of multi-agent systems, such as circular path-following \cite{Hu2019}, circular formation control \cite{Wen2019}, and plug-and-play maneuver in platooning \cite{Hu2009}. 

Successful implementations of DMPC require to solve distributed optimization problems in a real-time manner. However, due to the communication constraints and limitations in the AUV application, distributed optimization algorithms may suffer from noised iterations. Inexact distributed optimization algorithms can potentially deal with the errors resulted from noised iterations caused by unreliable or limited communication. Some works can be found in \cite{Magnusson2018}, \cite{Doan2020}, \cite{Q2016}, \cite{Devolder2014} and \cite{Nedelcu2014}. In \cite{Puy} and \cite{7402506}, the authors proposed an iteratively refining quantization design for distributed optimization and showed complexity upper-bounds on the number of iterations to achieve a given accuracy. 

In this paper, we aim to extend the quantized optimization algorithm proposed in \cite{7402506} to a real-time DMPC framework for multi-AUV systems with limited communication data-rates. The contributions of this paper can be summarized as 
\begin{itemize}
    \item Based on \cite{7402506}, we first present an novel approach to obtain an optimal quantization design to achieve the best sub-optimality subject to a limited communication data-rate. We establish a relationship between the quantization design and control performance.
    \item We propose a real-time DMPC framework based on a distributed optimization algorithm using the optimal quantization design. We further study the closed-loop properties of the system with the proposed approach.
    \item We apply the proposed approach to a case study of three AUVs with a real-time constraint on communication data-rates. The simulation results show the effectiveness of the DMPC with the optimal quantization design.
\end{itemize}
%
\section{Preliminaries}\label{section:prem}
\subsection{Notations}

Throughout this paper, we use the superscript $\star$ to indicate a variable as an optimal solution to an optimization problem. We use $k$ to denote the iteration step in an optimization algorithm, $l$ to denote a prediction step in an MPC problem, and $t$ to denote the time step of the closed-loop system. For a vector $x$, we use $\|x\|$ and $\|x\|_Q$ to denote the 2-norm and the weighted 2-norm, respectively. We use $ \operatorname{diag}(x)$ to indicate a diagonal matrix with the elements in diagonal induced by a vector $x$. We use $\operatorname{blkdiag}(X_1,\ldots,X_n)$ to indicate a block-diagonal matrix with elements in diagonal induced by matrices $X_1,\ldots,X_n$. We use $I$ to denote an identity matrix of appropriate dimension. Consider a distributed problem solved in a network of $M$ agents. The agents communicate according to a fixed undirected graph $G=(\mathcal{V}, \mathcal{E})$. The agents distribute according to the vertex set $\mathcal{V}=\{1, \cdots, M\}$ and exchange information according to the edge set $\mathcal{E} \subseteq \mathcal{V} \times \mathcal{V}$. If $(i, j) \in \mathcal{E}$, then agent $i$ is said to a neighbour of agent $j$ and $\mathcal{N}_{i}=\{j \mid(i, j) \in \mathcal{E}\}$ denotes the set of the neighbours of agent $i$. For a real number $z$, a uniform quantizer with quantization step-size $\Delta$ and mid-value $\bar{z}$ is defined by
\begin{equation}
    Q(z)=\bar{z}+\operatorname{sgn}(z-\bar{z}) \cdot \Delta \cdot\left\lfloor\frac{\|z-\bar{z}\|}{\Delta}+\frac{1}{2}\right\rfloor,
\end{equation}
where $\operatorname{sgn}(z-\bar{z})$ denotes a sign function and $\Delta=\frac{l}{2^{n}}$. The parameters $l$ and $n$ represent the length of the quantization interval and the number of bits sent through a quantizer at each iteration, respectively. The quantization interval is set as $\left[\bar{z}-\frac{l}{2}, \bar{z}+\frac{l}{2}\right]$. Note that if $z$ falls inside the quantization interval, then the quantization error satisfies $|z-Q(z)| \leq \frac{\Delta}{2}=\frac{l}{2^{n+1}}$.
\subsection{Parametric Distributed Optimization Problem}

Consider the following parametric distributed optimization problem $\mathbb{P}^p(\zeta^t)$:
\begin{subequations}\label{problem:PDproblem}
    \begin{align}
         &\min _{z, z_{\mathcal{N}_{i}}}  f\left(z, \zeta^{t}\right)=\sum_{i=1}^{M} f_{i}\left(z_{\mathcal{N}_{i}}, \zeta_{i}^{t}\right), \\
        \text{s.t. } & z_{i} \in \mathbb{C}_{i},\; z_{i}=F_{j i} z_{\mathcal{N}_{j}}, j \in \mathcal{N}_{i}, \\
        & z_{\mathcal{N}_{i}}=E_{i} z,\; i=1,2, \cdots, M, 
    \end{align}
\end{subequations}
where $z_i$ denotes the local variable, $z_{\mathcal{N}_{i}}$ denotes the concentration of the local variable $z_j$ where $j\in \mathcal{N}_{i}$ and $z=\left[z_{1}, \cdots, z_{M}\right]^{\top}$ denotes the global variable. The matrix $E_i$ selects $z_{\mathcal{N}_{i}}$ from the global variable $z$. The matrix $F_{ji}$ selects the local variable $z_i$ from $z_{\mathcal{N}_j}$. Furthermore, the local constraints are $z_{i} \in \mathbb{C}_{i} \subseteq \mathbb{R}^{{m_{i}}}$, where $m_i$ is the size of the local variable vector and $\mathbb{C}_i$ is a convex set, for $i=1,\cdots, M$. $\zeta^t_i$ is a time-varying parameter that does not change the convexity of the problem. 
\begin{assumption} \label{ap1}
The local cost function $f_i(\cdot)$ has a Lipshitz continuous gradient with respect to a Lipshitz constant $L_i$. The global cost function $f(\cdot)$ is strongly convex with a convexity modulus $\sigma_{f}$.
\end{assumption}

Note that if Assumption \ref{ap1} holds, then the global cost function $f(\cdot)$ has a Lipschitz continuous gradient with the Lipschitz constant $L_{max}$, where $L_{max}:=\max_{1<=i<=M} L_i$. 

\subsection{Parametric Distributed Optimization Algorithm with Warm-starting and Progressive Quantization Refinement}

\begin{algorithm}[h]
\textbf{Require:} Give $ z^{0,0}$, $K$, $C_{\alpha}$ and $C_{\beta}$, $(1-\gamma)<\kappa<1$ where $\gamma = \frac{\sigma_f}{L}$ and $ \tau<\frac{1}{L}$. \\
\textbf{\textbf{for}} $t = 0,1,\cdots$ \textbf{\textup{do}} \\
 1. Initialize $C_{\alpha}^{t}=C_{\alpha}$, $C_{\beta}^{t}=C_{\beta}$, $ \hat{z}_{i}^{t,-1}=z^{t,0}_i$ and $\hat{\nabla}f_{i}^{t,-1}=\nabla f_{i}(\operatorname{Proj}_{\mathbb{C}_{\mathcal{N}_{i}}}(z_{\mathcal{N}_{i}}^{t,0}))$; \\
\textbf{\textup{for}} $k = 0,1,\cdots,K$ \textbf{\textup{do}} \\
\textbf{\textup{for}} $i = 1,\cdots,M$ \textbf{\textup{do in parallel}} \\
2. Update quantizer $Q_{\alpha, i}^{t,k}: l_{\alpha, i}^{t,k}=C^{t}_{\alpha} \kappa^{k}$ and $\bar{z}_{\alpha, i}^{t,k}=\hat{z}_{i}^{t,k-1}$;\\
3. Quantize local variable $\hat{z}_{i}^{t,k}=Q_{\alpha, i}^{t,k}\left(z_{i}^{t,k}\right)$; \\
4. Send $\hat{z}^{t,k}_i$ to agent $j$ for all $j\in\mathcal{N}_i$; \\
5. Compute the projection: $ \tilde{z}_{\mathcal{N}_{i}}^{t, k}=\operatorname{Proj}_{\mathbb{C}_{\mathcal{N}_{i}}}\left(\hat{z}_{\mathcal{N}_{i}}^{t,k}\right)$; \\
6. Compute $\nabla f^{t,k}_i = \nabla f \left(\tilde{z}_{\mathcal{N}_{i}}^{t, k}
\right)$; \\
7. Update quantizer $Q_{\beta, i}^{t,k}: l_{\beta, i}^{t,k}=C_{\beta}^t \kappa^{k}$ and $\bar{\nabla} f_{\beta, i}^{t,k}=\hat{\nabla} f_{i}^{t,k-1}$;\\
8. Quantize local gradient $\hat{\nabla} f_{i}^{t,k}=Q_{\beta, i}^{t,k}\left(\nabla f_{i}^{k}\right)$;\\
9. Send $\hat{\nabla} f^{t,k}_i$ to agent $j$ for all $j\in\mathcal{N}_i$;\\
10. $z^{t,k+1}_i = \mathrm{Proj}_{\mathbb{C}_{i}} \left(z^{t,k}_i - \tau \sum_{j \in \mathcal{N}_{i}} F_{j i} \hat{\nabla} f_{j}^{t,k} \right)$; \\
\textbf{\textup{end}} \textbf{\textup{end}}\\
11. Warm-start update: $z_i^{t+1,0} = z^{t,K+1}_i$ for $i=1,\cdots, M$;\\
12. \textbf{\textup{Return: }} $z^{t,K+1}_i$ for $i=1,\cdots, M$.\\
\textbf{\textup{end}}
\caption{Parametric Distributed Optimization with Warm-starting and Progressive Quantization Refinement}
\end{algorithm}

We now introduce Algorithm 1 is introduced to solve $\mathbb{P}^p(\zeta^t)$. In Algorithm 1, $Q_{\alpha, i}^{t,k}$ and $Q_{\beta, i}^{t,k}$ are two uniform quantizers. The subscripts $\alpha$ and $\beta$ indicate that they are used to quantize local variables and local gradients, respectively. The quantizers are refined at each iteration by shrinking the size of their quantization intervals according to $ l_{\alpha, i}^{t, k}= C^{t}_{\alpha} \kappa^{k}; l_{\beta, i}^{t,k}= C^t_{\beta} \kappa^{k}$ where $C^t_\alpha$ and $C^t_\beta$ are the initial quantization intervals and $\kappa$ is a shrinkage constant. The mid-values of the quantizers are updated according to $\bar{z}_{\alpha, i}^{t,k}=\hat{z}_{i}^{t,k-1}$ and $\bar{\nabla} f_{\beta, i}^{t,k}=\hat{\nabla} f_{i}^{t,k-1}$. The quantized values are designated by $\hat{\cdot}$ while output of the projection steps are designated by $\tilde{\cdot}$. The operation $\operatorname{Proj}_{\mathbb{C}}(v):=\operatorname{argmin}_{\mu \in \mathbb{C}}\|\mu-v\|$ represents the projection of any point $v \in \mathbb{R}^{n_{v}}$ on the set $\mathbb{C}$.

\subsection{Complexity Upper-bound for Algorithm 1}

There exist an error upper bound on the sub-optimality of the solutions given by Algorithm 1 at all time steps if the following assumptions are made.

\begin{assumption}\label{ap3}
For all $t \geq 0$, the solutions to $\mathbb{P}^p(\zeta^t)$ satisfy
\begin{equation} \label{solution upper bound}
    \left\|z^{\star}\left(\zeta^{t}\right)-z^{\star}\left(\zeta^{t+1}\right)\right\| \leq \rho.  
\end{equation}
\end{assumption}

\begin{assumption}\label{ap4}
The initial solution to $\mathbb{P}^p(\zeta^0)$ satisfies 
\begin{equation}
    \left\|z^{0}\left(\zeta^{0}\right)-z^{\star}\left(\zeta^{0}\right)\right\| \leq \epsilon.
\end{equation}
\end{assumption}

\begin{assumption}\label{ap5}
The initial quantization intervals $C_\alpha$ and $C_\beta$ satisfy  
\begin{subequations}
    \begin{align}
        a_{1} (\epsilon+\rho)+a_{2} \frac{C_{\alpha}}{2^{n+1}}+a_{3} \frac{C_{\beta}}{2^{n+1}} & \leq \frac{C_{\alpha}}{2},\\
        b_{1} (\epsilon+\rho)+b_{2} \frac{C_{\alpha}}{2^{n+1}}+b_{3} \frac{C_{\beta}}{2^{n+1}} & \leq \frac{C_{\beta}}{2}, 
    \end{align}
\end{subequations}
where the coefficients $a_1, a_2, a_3$ and  $b_1, b_2, b_3$ can be found in Appendix.
\end{assumption}

\begin{theorem}[\cite{7402506}] \label{tm1}
Consider Assumption 1-4 hold and the number of iteration $K$ at all time steps $t \geq 0$ satisfies
\begin{equation}
    K \geq \ceil*{\log _{\kappa} \frac{\epsilon(1-\kappa)}{\rho+\delta+(1-\kappa)(\epsilon+\delta)}}-1,
\end{equation}
where $\delta=\frac{\kappa\left(C_{1}+\sqrt{2 L} C_{2}\right)}{L(\kappa+\gamma-1)(1-\gamma)}, C_{1}=\frac{M \sqrt{m_{max}}\left(L_{\max } d C_{\alpha}+\sqrt{d} C_{\beta}\right)}{2^{n+1}}$, and $C_{2}=\frac{\sqrt{2}}{2} \cdot \frac{M \sqrt{m_{max}} C_{\alpha}}{2^{n+1}}$. 
The parameter $m_{max}:=\max_{1<=i<=M} m_i$ is the largest size of local variables, $d$ is the degree of the graph and $L_{max}$ is the largest Lipshitz constant of the gradient of the local cost $f_i$, i.e., $L_{max}:=\max_{1<=i<=M} L_i$, and $L$ is the Lipshitz constant of the gradient of the global cost $f$. Then, at all time steps $t \geq 0$, the solution given by Algorithm 1 satisfies
\begin{equation}\label{error upper bound}
    \left\|z^{K+1}\left(\zeta^{t}\right)-z^{\star}\left(\zeta^{t}\right)\right\| \leq \epsilon.
\end{equation}
\end{theorem}

\section{Optimal Quantization Design}\label{section:Quantization Design}

\subsection{Optimal Quantization Design Formulation}

According to Theorem \ref{tm1}, we know that the required number of iterations $K$ to upper-bound the sub-optimality of the solutions given by Algorithm 1 at a certain level $\epsilon$. This upper-bound also establishes a relationship between the solution accuracy and the parameters of the quantizers $C_{\alpha}$, $C_{\beta}$, the quantization shrinkage constant $\kappa$ and the number of bits of the quantizers $n$. This is further related to the number iterations $K$ in Algorithm 1 given the real-time communication constraint $T \geq \ceil{nK} $ at each time step $t$, where $T$ is a communication data-rate. In the following, the set of the parameters $(\kappa,n,K,C_{\alpha},C_{\beta})$ is referred to as a quantization design. The relationship shown in (\ref{error upper bound}) is used to compute the optimal quantization design yielding the lowest sub-optimality upper bound for a fixed $T$.
\begin{subequations}\label{problem:UpperBound}
	\begin{align}
    	& && \min _{\kappa,n,K,C_{\alpha},C_{\beta},\epsilon} \quad \epsilon,\\
    	&\text{s.t. } &&  K \geq\left\lceil\log _{\kappa}\frac{\epsilon(1-\kappa)}{\rho+\delta+(1-\kappa)(\epsilon+\delta)}\right\rceil-1, \label{st1}\\
    	& && T \geq \ceil{nK}, \\
    	& && a_{1} (\epsilon+\rho)+a_{2} \frac{C_{\alpha}}{2^{n+1}}+a_{3} \frac{C_{\beta}}{2^{n+1}} \leq \frac{C_{\alpha}}{2}, \label{st2}\\
    	& && b_{1} (\epsilon+\rho)+b_{2} \frac{C_{\alpha}}{2^{n+1}}+b_{3} \frac{C_{\beta}}{2^{n+1}} \leq \frac{C_{\beta}}{2}.\label{st3}
	\end{align}
\end{subequations}

It can be seen that the optimization problem in \eqref{problem:UpperBound} is non-convex and computationally challenging for existing solvers. Alternatively, we propose a method to compute an approximation the optimal solution for problem \eqref{problem:UpperBound} by solving a set of convex sub-problems where each sub-problem is specified with a fixed ($\kappa,n,K$) configuration from the set $\{(\kappa,n,K)\mid T \geq \ceil{nK} , 1 - \gamma \leq \kappa \leq 1, \kappa \ \mathbf{rem} \ {\lambda } = 0 \}$. The operation $\kappa \ \mathbf{rem} \ {\lambda }$ denotes for the remainder of $\frac{\kappa}{\lambda}$ where $\lambda$ is a small and positive constant. The constant $\lambda$ indicates the resolution level at which $\kappa$ is selected. With fixed $(\kappa,n,K)$, the sub-problem given in (9) is convex and solvable by existent solvers.
\begin{align}\label{problem:UpperBoundInterger}
	& \min _{C_{\alpha},C_{\beta},\epsilon} \quad \epsilon,\\
	\text{s.t. } & \text{(\ref{st1}),(\ref{st2}),(\ref{st3})} \nonumber.
\end{align}

The (${\kappa^{a\star},n^{a\star},K^{a\star},C_{\alpha}^{a\star},C_{\beta}}^{a\star}$) configuration associated with the smallest sub-optimality upper bound $\epsilon^{a\star}$ out of the solutions of all sub-problems is said to be the best approximation of the optimal quantization design, namely the optimal solution of problem (\ref{problem:UpperBound}).

\begin{remark}
Since problem (\ref{problem:UpperBound}) is solved off-line, all possible $(n,K)$ pairs can be considered for solving (\ref{problem:UpperBoundInterger}) for a fixed $T$. Although $\kappa$ is a continuous variable, a high resolution (a small $\lambda$) can be chosen to uniformly select $\kappa$ from the range $(1-\mu,1)$. Therefore, by considering all possible $(n,K)$ such that $\ceil{nK} \leq T$ and choosing $\lambda$ to be a small and positive number, the solution (${\kappa^{a\star},n^{a\star},K^{a\star},C_{\alpha}^{a\star},C_{\beta}}^{a\star}$) provides a close estimation of the optimal quantization design.
\end{remark}

\subsection{Example}
Consider the following parametric distributed quadratic optimization problem:
\begin{subequations}\label{problem:QPproblem}
    \begin{align}
        &  \min _{z, z_{\mathcal{N}_{i}}} && \sum_{i=1}^{M} z_{\mathcal{N}_{i}}^{T}H_{i}z_{\mathcal{N}_{i}} +  (\zeta^t_i)^{T} h_i z_{\mathcal{N}_{i}}, \\
        &    \text{s.t. } && G_i z_i \leq h_i,\; z_{i}=F_{j i} z_{\mathcal{N}_{j}}, j \in \mathcal{N}_{i},\\
         &   && z_{\mathcal{N}_{i}}=E_{i} z,\; i=1,2, \cdots, M.
    \end{align}
\end{subequations}

\begin{figure}[t]
 \centering
 \subfigure[]{\includegraphics[width=0.49\hsize]{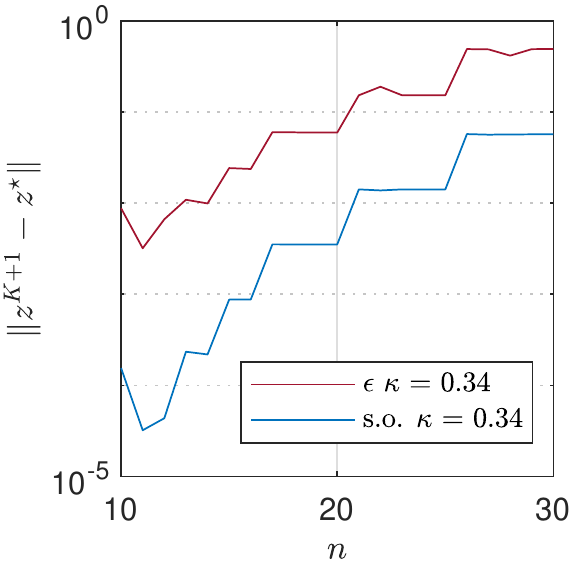}\label{bound_a}}
 \subfigure[]{\includegraphics[width=0.49\hsize]{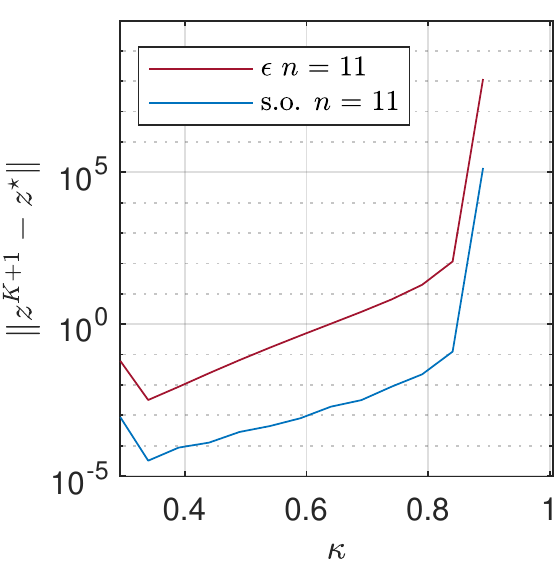}\label{bound_b}}
\caption{Complexity upper-bound in (5) (red curves) v.s. true sub-optimality (blue curves); (a) Fixing $\kappa$ at the optimal value 0.34 and varying the number of bits $n$; (b) Fixing the number at the optimal value 11 and varying $\kappa$.}
\label{img error upper bound}
\end{figure}

In this example, we randomly generate a connected graph with $6$ agents. The degree of the graph is equal to $2$. Each agent has 2 local variables. The matrix $H_i$ is set to be a randomly generated positive definite matrix and the vector $h_i$ is also randomly generated, for $i=1,\cdots,M$. The time-varying parameter $\zeta^t_i$ is uniformly sampled from a constant interval at each time step $t$. The polytopic constraint $G_i z_i \leq h_i$ are also randomly generated but with the guatrantee that more $50\%$ optimization variables can hit the constraints. Furthermore, we set the total number of bits transmitted at each step $t$ to be $T = 100$ bits. For this example, the parameters required to  compute of the optimal quantizaton design are $L = 21.99$, $L_m = 16.54$, $\sigma_{f} = 15.93$, $\gamma = 0.72$ and $\rho = 8.42$ obtained by sampling. We compute an approximation of the optimal quantizaton design using the approach presented in Section III.A and compare the result with the true sub-optimality achieved by Algorithm 1.

Fig. \ref{img error upper bound} presents a comparison between the sub-optimality upper-bound $\epsilon$ and the true sub-optimality achieved by the optimization algorithm using different quantization configurations, represented by the red curves and the blue curves, respectively. We can see that the red and blue curves show the same trend. The upper-bound $\epsilon$ is tight to the true sub-optimality. In Fig. \ref{bound_a}, we fix $\kappa$ and compare $\epsilon$ with the true sub-optimality using different $n$ ($K$ is set to be $\floor{T/n}$). For both cases (the red and blue curves), we observe that the lowest sub-optimality is achieved at $n=11$. In Fig. \ref{bound_a}, we fix the number of bits $n$ ($K$ is set to be $\floor{T/n}$) and compare $\epsilon$ with the true sub-optimality using different $\kappa$. For both cases (the red and blue curves), we observe that the lowest sub-optimality is achieved at $\kappa=0.34$.
%
%

\section{DMPC with Limited Communication Data-rates for Multiple AUVs}\label{section:AUV Formation Control}

We first introduce a nonlinear model of an AUV and a discrete-time linear model. Then, we present a DMPC formulation for formation control of multiple AUVs. Furthermore, we propose a real-time framework for solving the DMPC problem subject to communication constraints. The framework consists of two stages: an offline stage and an online stage. In the offline stage, we use the method presented in Section III. A to find the optimal quantization design. In the online stage, we use the distributed optimization algorithm with in Algorithm 1 to solve the DMPC  problem with the off-line computed optimal quantization design. Finally, we briefly analyze the closed-loop properties. 

\subsection{AUV Dynamic Model}

\begin{figure}[t]
\centering
\includegraphics[width = \hsize]{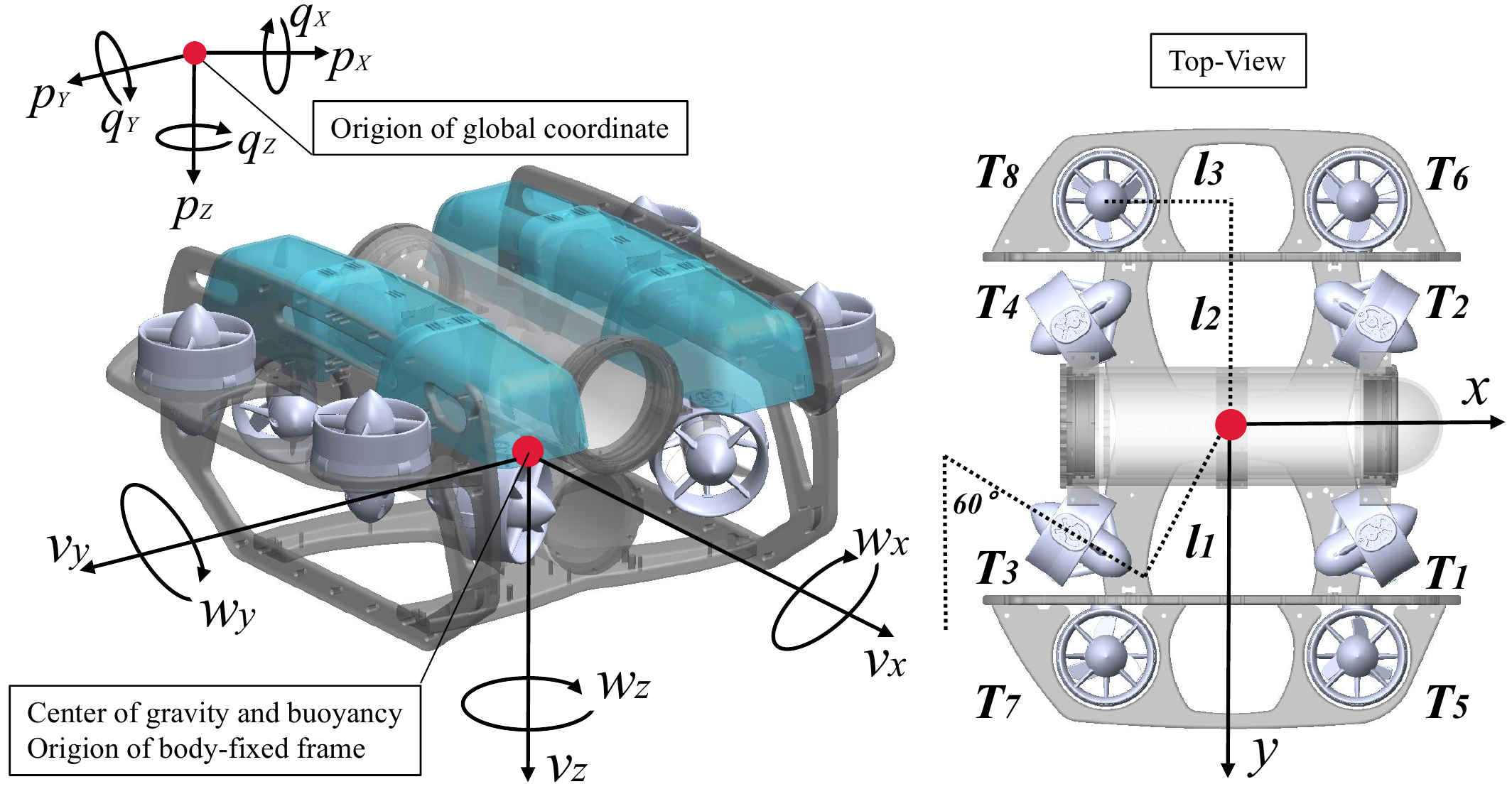}
\caption{AUV Schematic.}
\label{fig:AUV}
\end{figure}

In Fig. \ref{fig:AUV}, the schematic of a BlueROV2 underwater vehicle \footnote{$https://bluerobotics.com/store/rov/bluerov2/$} is shown. 
To model this AUV, let us first define the states and the control inputs of the system as $x=[\eta, \nu]^\top$ and $u=[T_1,T_2,T_3,T_4,T_5,T_6,T_7,T_8]^\top$, respectively. We use the vector $\eta=[p_X,p_Y,p_Z,q_X,q_Y,q_Z]^\top$ to denote the position and orientation components in the global coordinate, and the vector $\nu=[v_x,v_y,v_z,w_x,w_y,w_z]^\top$ to denote the 
linear and angular velocities in the body-fixed coordinate. It is assumed the center of gravity (C.G.) and the center of buoyancy (C.G.) collocate in the same horizontal plane as the horizontal thrusters $T1,T2,T3,$ and $T4$ as labeled in Fig. \ref{fig:AUV}. The nonlinear AUV model can be expressed as
\begin{align}\label{nonlinear model}
&\dot{x}= 
\begin{bmatrix}
J(\eta)\nu\\
M^{-1}\left(\tau_{c}(u)-C(\nu) \nu-D(\nu) \nu-g(\eta)\right)
\end{bmatrix},
\end{align}
where $J(\eta)=\text{diag}(R,W)$ with the transformation matrices $R$ and $W$ mapping the linear and angular velocities from the body-fixed frame to the global coordinate. The matrix $M = M_r + M_a$ represents the total mass. $M_r = \operatorname{diag} ([m,m,m,m,m,m])$ is the rigid body mass and $ M_a =  - \operatorname{diag} ([X_{\dot{v}_z},Y_{\dot{v}_y},Z_{\dot{v}_z},K_{\dot{\omega}_x},M_{\dot{\omega}_y},N_{\dot{\omega}_z}])$ is the added mass associated with linear and angular velocities. Moreover, $\tau_c (u)$ is the control input and can be formulated as follows:
\begin{equation*}
 \tau_c (u) = \tau u = 
\begin{bsmallmatrix}
\sin\theta &\sin\theta & \sin\theta & \sin\theta & 0 & 0 & 0 & 0\\
-\cos\theta & \cos\theta & \cos\theta & -\cos\theta & 0 & 0 & 0 & 0 \\
0 & 0 & 0 & 0 & 1 & 1 & 1 & 1 \\
0 & 0 & 0 & 0 & l_3  & -l_3 & l_3 & -l_3 \\
0 & 0 & 0 & 0 & l_2 & l_2  & -l_2 & -l_2 \\
l_1  & -l_1 & l_1 & -l_1 & 0 & 0 & 0 & 0
\end{bsmallmatrix}
u,
\end{equation*}
where $\theta = \frac{\pi}{3}$ is the tilt angle of the horizontal thrusters based on the AUV structure in Fig. \ref{fig:AUV}, $l_1$ represents the distance between the horizontal thrusters and the C.G., $l_2$ represents the distance between the vertical thrusters and the $X$-axis of the body-frame, and $l_3$ represents the distance between the vertical thrusters and the $Y$-axis of the body-frame. $C(\nu) \nu$, $D(\nu) \nu$ and $g(\eta) $ correspond to the Coriolis force, damping force and gravitational/buoyant force, which can be found in Appendix. We refer to \cite{Yang2019} for further details.

For the nonlinear AUV model in \eqref{nonlinear model}, we can use the first-order Taylor expansion at the point ($x_n = 0 $ and $u_n = 0$) to obtain a discrete-time linear model as follows:
\begin{equation}\label{local model}
    {x}(t+1) = Ax(t) + B u(t),
\end{equation}
where $ A =
\begin{bsmallmatrix}
I & \Delta t I \\
0 & I
\end{bsmallmatrix} $, $B = 
\Delta t M^{-1} \tau $. $\Delta t$ is the sampling time. 
\begin{remark}
    For the linear model in \eqref{local model}, the matrix pair~$(A,B)$ is controllable.
\end{remark}
\subsection{DMPC Formulation}
We now formulate the DMPC optimization problem for formation control of multiple AUVs. Under a leader-follower framework, the agent $i=1$ is assigned as the leader and agent $ i \geq 1$ as the followers. The agents share information according to a fixed communication graph $\mathcal{G}$. The agents connected by the edges in the edge set $\mathcal{E}$ keep their relative distances. Since the states of the agents are not coupled in the dynamics, the global system is also controllable as long as the individual agents are controllable. In general, the DMPC problem $\mathbb{P}^s(x,x_r,u_r)$ can be formulated as follows.
\begin{subequations}\label{soft DMPC problem}
\begin{alignat}{3}
&\min_{\substack{\bar{x}_i(l),\bar{u}_i(l)}} && \;\; \sum_{i=1}^{M} \sum_{l=0}^{N-1} \ell_{i}\left(\bar{x}_{i}(l)-x_{r_i}, \bar{u}_{i}(l)-u_{r_i}\right)\nonumber \allowdisplaybreaks\\
& &&\hspace{-10mm} + \sum_{(i,j)\in\mathcal{E}}\sum_{l=0}^{N-1} \ell^d_{ij}(G_{ij} [\bar{x}_{i}^{\top}(l), \bar{x}^{\top}_j(l)]^{\top} - G_{ij} [x_{r_i}^{\top}, x^{\top}_{r_j}]^{\top}) \nonumber \allowdisplaybreaks \\
& && \hspace{-10mm} +\sum_{i=1}^{M} \ell_{i}^{f}\left(\bar{x}_{i}(N)-x_{r_i}\right) \label{eq:total DMPC cost function} ,\allowdisplaybreaks \\
&\text{s.t. } && \bar{x}_i(0)=x_i(t), \label{cons1} \allowdisplaybreaks\\
& && \bar{x}_i(l+1)=A \bar{x}_i(l)+ B \bar{u}_i(l), \allowdisplaybreaks\\
&  && G_{u_i} \bar{u}_i(l) \leq h_{u_i},\label{stage input constraint} \allowdisplaybreaks\\
&  && G_{x_i} \bar{x}_i(l) \leq h_{x_i}, \label{stage state constraint} \allowdisplaybreaks\\
&  && \bar{x}_{i}(N) \in \mathcal{E}_{f_{i}}^{s}\left(x_{r_{i}}\right), \;\;\forall i=1, \cdots, M,\label{terminal constraint}
\end{alignat}
\end{subequations}
where $x_i(t)$ is a measured system state of~\eqref{local model} at a time instant $t$. $\bar{x}_i(l)$ and $\bar{u}_i(l)$ denote the state and control input of agent $i$, respectively. $x_{r_i}$ and $u_{r_i}$ denote the state and input reference of agent $i$. The polytopic constraints \eqref{stage input constraint} and \eqref{stage state constraint} denote the local input constraint and the local state state constraint, respectively. $\mathcal{E}^s_{f_i}(x_{r_i})$ represents the local terminal constraint. 

The stage cost is defined as $\ell_i (\bar{x}_{i}(l)-x_{r_i}, \bar{u}_{i}(l)-u_{r_i} ) = \|\bar{x}_{i}(l)-x_{r_i}\|_{Q_i}^2+\| \bar{u}_{i}(l)-u_{r_i}\|_{R_i}^2$. The cost function for formation control is defined as $\ell^d_{ij}(G_{ij} [\bar{x}_{i}^{\top}(l), \bar{x}^{\top}_j(l)]^{\top} - G_{ij} [x_{r_i}^{\top}, x^{\top}_{r_j}]^{\top}) = \|G_{ij} [\bar{x}_{i}^{\top}(l), \bar{x}^\top_j(l)]^{\top}- G_{ij} [x_{r_i}^{\top}, x^\top_{r_j}]^{\top}\|_{S_{ij}}^2$. The terminal cost is defined as $\ell_{i}^{f} \left(\bar{x}_{i}(N)-x_{r_i}\right)= \|\bar{x}_i(N)-x_{r_i}\|_{P_i}^2 $. The weighting matrices $Q_i$, $R_i$, $P_i$, and $S_{ij}$ are set to be positive definite matrices. $P = \operatorname{blkdiag}(P_i,\ldots,P_M)$ can be obtained by solving the algebraic Riccati equation with $ Q = \operatorname{blkdiag}(Q_i,\ldots,Q_M)$ and $R = \operatorname{blkdiag}(R_i,\ldots,R_M)$.

In the DMPC problem \eqref{soft DMPC problem}, formation maintenance is achieved by penalizing the difference between the relative distances and the reference relative distances $G_{ij} [\bar{x}_{i}^{\top}(l), \bar{x}^\top_j(l)]^{\top}- G_{ij} [x_{r_i}^{\top}, x^\top_{r_j}]^{\top}$. Given a set of reference setpoints $(p_{r_{X_i}}, p_{r_{Y_i}}, p_{r_{Z_i}})$, for agent $i=1, \cdots, M$, the reference setpoint for the leader, i.e., agent $1$, is set to be $x_{r_1} = (p_{r_{X_1}}, p_{r_{Y_1}}, p_{r_{Z_1}})$. The reference setpoint for the followers, i.e., agent i for $i>1$, is set to be $x_{r_j} = (p_{r_{X_i}}-d_{r_{X_{ij}}}, p_{r_{Y_i}}-d_{r_{Y_{ij}}}, p_{r_{Z_i}}-d_{r_{Z_{ij}}})$. Among them, $d_{r_{X_{ij}}}, d_{r_{Y_{ij}}},$ and $d_{r_{Z_{ij}}}$ are given references for the relative distances in the $X-$axis, $Y-$axis and $Z-$axis, respectively. The matrix $G_{ij}$ is used to select $p_{X_i}-p_{X_j}, p_{Y_i} -p_{Y_j}, p_{Z_i} -p_{Z_j}$ from the concatenated vector $[\bar{x}_{i}^{\top}(l), \bar{x}^\top_j(l)]^{\top}$. Similarly, $d_{r_{X_{ij}}}, d_{r_{Y_{ij}}},d_{r_{Z_{ij}}}$ can be selected by using the matrix $G_{ij}[x_{r_i}^{\top}(l), x^\top_{r_j}(l)]^{\top}$. 

\begin{remark}\label{remark:DMPC reformulation}
$\mathbb{P}^r(x,x_r,u_r)$ can be reformulated as a distributed QP in the form of \eqref{problem:QPproblem}. For each agent $i$, the optimization variable is set to be $z_i = [\bar{x}_i^{\top}(0),\cdots,\bar{x}_i^{\top}(N),\bar{u}_i^{\top}(0),\cdots,\bar{u}_i^{\top}(N-1)]^{\top}$. The local formation cost function can be written as $\|[x_{i}^{\top}(l), x^\top_j(l)]^{\top}-[x_{r_i}^{\top}, x^\top_{r_j}]^{\top}\|_{\tilde{S}_{ij}}^2$ with $\tilde{S}_{ij} = G_{ij}^\top S_{ij} G_{ij}$. The augmented block-diagonal matrix $\textup{H}_i$ can be built with $Q_i$, $R_i$, $\tilde{S}_{ij}$ and $P_i$. The local constraint $\mathbb{C}_i$ can be obtained by reformulating (13b)-(13f). 
\end{remark}

We now summarize the real-time DMPC framework for solving and implementing the DMPC problem in \eqref{soft DMPC problem} with a limited communication data-rate~$T$ in Algorithm 2. In the off-line stage, the optimal quantization design $(\kappa,n,K,C_{\alpha},C_{\beta})$ is obtained by solving the optimization problem in \eqref{problem:UpperBound}. Then, the distributed optimization algorithm in Algorithm 1 with the optimal quantization design obtained from the off-line stage is used to solve the DMPC problem in \eqref{soft DMPC problem} in the on-line stage.


\begin{algorithm}[t]
\textbf{\textup{Off-line stage:}}\\
1. Consider $\mathbb{P}^s(x(t),x_r,u_r)$ in (\ref{soft DMPC problem}), reform it to be (\ref{problem:QPproblem}) and find $M, L, L_m, d,$ $ m, \gamma$, and $\tau$. Find $\rho$ specified by (\ref{solution upper bound}) through sampling;\\
2. Given $M, T, L, L_m, d, m, \gamma, \tau, \text{and } \rho$, solve problem (\ref{problem:UpperBound}), via solving the sub-problems in (\ref{problem:UpperBoundInterger}), to calculate the optimal quantization design $(\kappa^\star,n^\star,K^\star,C_{\alpha}^\star,C_{\beta}^\star)$ and the corresponding optimal error bound $\epsilon^\star$; \\
3. Given the initial state of the system $x_{i}(0) $, for $i=1,\cdots,M$, calculate ${z}^{0,0}$ s.t. $\|{z}^{0,0}-{z}^{0,\star}\| \leq \epsilon^\star$. \\
\textbf{\textup{On-line stage:}}\\
\textbf{\textup{for}} $t = 0,1,\cdots$ \textbf{\textup{do}} \\
4. Measure $x_i(t)$;\\
5. Execute step (1)-(10) in Algorithm 1 using the quantization design $(\kappa^\star,n^\star,K^\star,C_{\alpha}^\star,C_{\beta})^\star$ to solve $\mathbb{P}^s(x(t),x_r,u_r)$ in (13) with the initial solution ${z}^{t,0}$;\\
6. Update the initial solution for $t+1$: ${z}^{t+1,0} = {z}^{t, K+1}$ ;\\
7. Extract $u(0)^{t,K+1}_{i}$ from ${z}^{t,K+1}_i$ and apply it to agent $i$, $\forall{i\in\{1,\cdots,M\}}$; \\
\textbf{\textup{end}}
\caption{Real-time DMPC framework}
\end{algorithm}

\begin{remark}
The parameter $\rho$ introduced in Assumption \ref{ap3} is determined by simulation. For each sample, the system is driven from a random initial state in the set of admissible initial states to the origin using Algorithm 2. The parameter $\rho$ is set to be the maximum distance between the optimal solutions of any two sampled states. 
\end{remark}

\subsection{Closed-loop Analysis for the DMPC in (13) with a Limited Communication Data-rate}



In this section, we analyze the closed-loop properties for the system with the proposed DMPC formulation in (13) following the implementation steps in Algorithm 2 in the following corollary.

\begin{corollary}
    Consider the distributed MPC problem in (13) and Algorithm 2. Given the total number of bits $T$, the closed-loop system \eqref{local model}-\eqref{soft DMPC problem} with the real-time DMPC framework in Algorithm 2 is recursively feasible and input-to-state stable (ISS).
\end{corollary}

\begin{proof}
    \textit We first discuss the feasibility of the solution ${z}^{t,K+1}_i$ returned by Algorithm 2. Since each agent has only local dynamical, state and input constraints, the re-projection step (Step 5 in Algorithm 1) guarantees that the sub-optimal solution ${z}^{t,K+1}_i$ returned by Algorithm 2 is feasible for all $t\geq 0$. We then study the stability property. In the closed-loop system, the computational error (sub-optimality) of the control input $u(0)^{t,K+1}_{i}$ generated by Algorithm 2 can be considered as a disturbance $w$ in the closed-loop system. We know that this induced disturbance is bounded by a constant determined by the sub-optimality level $\epsilon$. Due to the fact that the closed-loop system is uniformly continuous in $x$ and $w$, the ISS property of the closed-loop system is implied by following \cite[Theorem 4]{Limon2009}.
\end{proof}    

\section{Case Study: a Multi-AUV System}\label{section:Formation Control with sub-optimality}

In this section, we apply the DMPC framework with the optimal quantization design applied to the multi-agent system with three AUVs. For the agents $i \in \mathcal{V} = \{1,2,3\}$, agent 1 is assigned as the leader and agent 2, 3 as the followers. The agents share information according to a fixed undirected graph $\mathcal{G}$, where the edges are $\mathcal{E} = \{  (1,2 ), (1,3 ) \}$. Also, followers only keep relative distance with respect to the leader. Each agent can be modelled by the AUV model in \eqref{local model}. The model parameters of the AUV are given as follows: $m = 11$, $W = B = 107.8$ (with ballast), $X_{\dot{v}_z} = -2.8$, $Y_{\dot{v}_y}=-3$, $Z_{\dot{v}_z}=-3.2$, $K_{\dot{\omega}_x}=-0.05$, $M_{\dot{\omega}_y}=-1$, $N_{\dot{\omega}_z}=-0.3$, $X_{v_x|v_x|}|v_x|,Y_{v_y|v_y|}|v_y|,Z_{v_z|v_z|}|v_z|=10$ and $K_{\omega_x|\omega_x|}|\omega_x|,M_{\omega_y|\omega_y|}|\omega_y|,N_{\omega_z|\omega_z|}|\omega_z|=4$. 

The state constraints are considered for three AUVs as $p_X \in [5,-10]$, $p_Y,p_Z \in [5,-5]$, $q_X,q_Y,q_Z \in [\frac{\pi}{3},-\frac{\pi}{3}]$, $v_x,v_y,v_z \in [1,-1]$, and $\omega_x,\omega_y,\omega_z \in [\frac{\pi}{6},\frac{\pi}{6}]$. Furthermore, the input constraints are set as $T_1,T_2,T_3,T_4,T_5,T_6,T_7,T_8 \in [2,-2]$ for three AUVs.

The tracking reference setpoints are set as $(0,0,0)$, $(-2,1,0)$ and $(-2,-1,0)$ for three AUVs. Three AUVs are expected to maintain a formation with desired relative distances as follows: $p_{X_{1}}-p_{X_{2}} = 2$, $p_{Y_{1}}-p_{Y_{2}} = -1$, $p_{Z_{1}}-p_{Z_{2}} = 0$, $p_{X_{1}}-p_{X_{3}} = 2$, $p_{Y_{1}}-p_{Y_{3}} = 1$ and $p_{Z_{1}}-p_{Z_{3}} = 0$. The initial position of three AUVs are $(-8,0.5,-0.5)$, $(-7,1.5,1)$ and $(-8,-0.5,-1)$, respectively. A sampling time $\Delta t = 0.1$ s is considered and the communication data-rate is $T =100$ bits at every step $t$.  

\subsection{Closed-loop Simulation Results}

We first show simulation results to demonstrate the stability of the closed-loop multi-AUV system controlled with the proposed DMPC framework with the optimal quantization design. The agents are modeled by the discrete linear model \eqref{local model}. The DMPC with optimal quantization design has been implemented following the steps described in Algorithm 2. In Fig. 3(a), the position of agent 2 is measured as its distance from its reference setpoint $(-2,-1,0)$. It can be seen that agent 2 moves toward the reference setpoint slowly from $t = 0$ to $2$ because it prioritizes formation tracking first. The position of agent 2 then converges to the reference setpoint after $t = 10$. As shown in Fig. 3(b), the control inputs saturate between $t = 0$ and $t = 6$ corresponding to the process of formation control.

\begin{figure}[t]
 \centering
 \subfigure[]{\includegraphics[width=0.44\hsize]{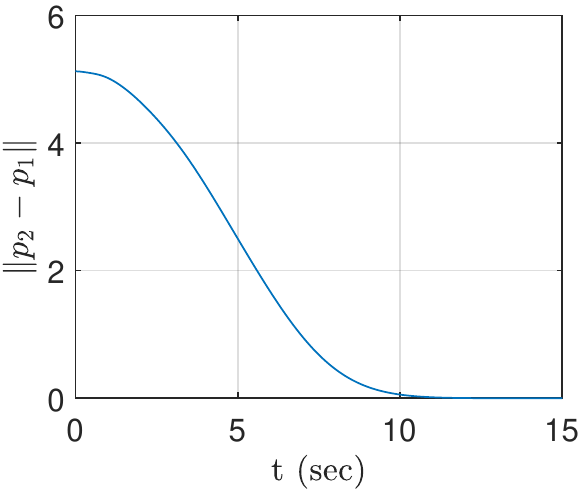}\label{converge_a}}
 \subfigure[]{\includegraphics[width=0.45\hsize]{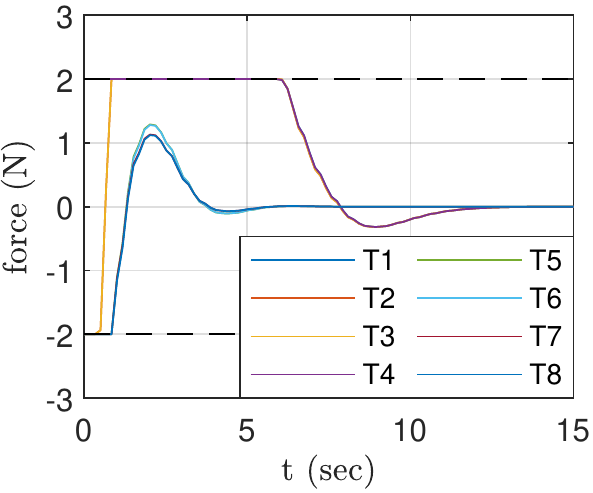}\label{converge_b}}
\caption{(a) Distance of agent 2 from its reference setpoint; (b) Control input of agent 2.}
\label{global error upper bound}
\end{figure}

\subsection{Control Performance with the Quantization Design}

We next show simulation results to compare control performance with two quantization designs. In this case, no terminal constraint is used in the DMPC framework in \eqref{soft DMPC problem}. The positive definite matrix $P_i$ in the terminal cost $\ell_i^f(x_i(N)-x_{r_i})$ is replaced by $Q_i$. The DMPC problem was implemented following the steps described in Algorithm 2 to control the system of three AUVs with non-linear model \eqref{nonlinear model}. By solving the optimization problem~\eqref{problem:UpperBound}, we can obtain the optimal quantization design to be $(\kappa=0.51,n=20,K=5,C_\alpha =65.85,C_\beta = 66.23)$. The sub-optimal quantization design is chosen as $(\kappa=0.95,n=30,K=3,C_\alpha =693.51,C_\beta = 693.72)$. For comparison of the two different quantization designs, random state disturbances are sampled from the intervals $[-0.1,0.1]$ and $[-0.03,0.03]$ with uniform distribution at each time step $t$. Then, these sampled disturbances are added to the linear states $p_X,p_Y,p_Z$ and angular states $q_X,q_Y,q_Z$, respectively.

\begin{figure}[t]
 \centering
 \subfigure[]{\includegraphics[width=0.49\hsize]{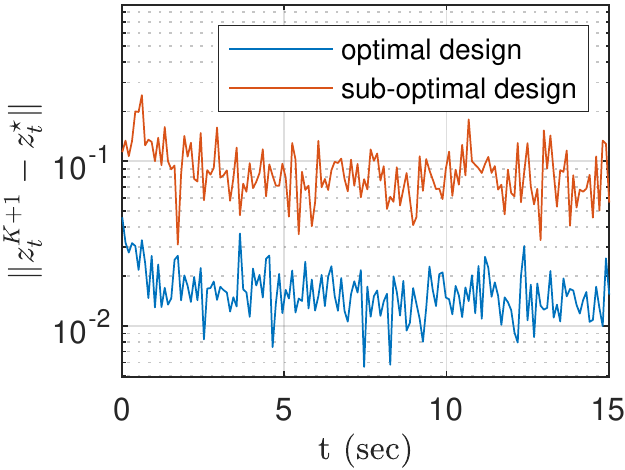}\label{perform_a}}
 \subfigure[]{\includegraphics[width=0.475\hsize]{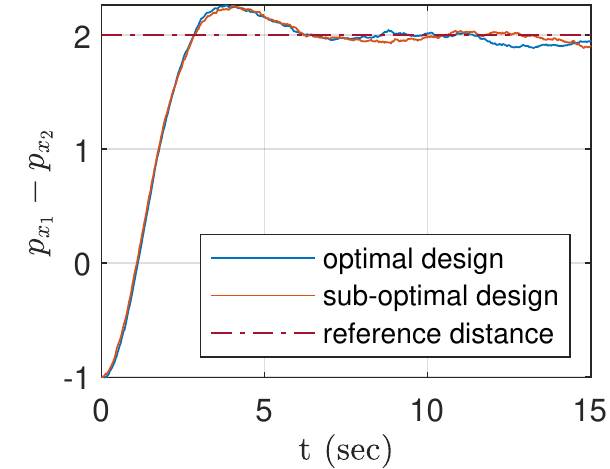}\label{perform_b}}\\
 \subfigure[]{\includegraphics[width=0.49\hsize]{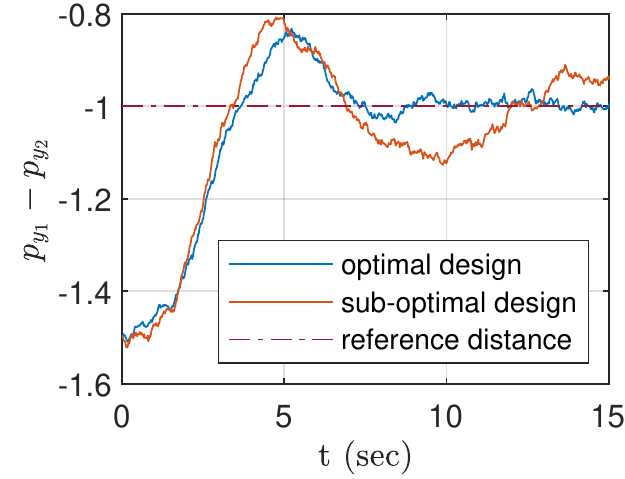}\label{perform_c}}
 \subfigure[]{\includegraphics[width=0.48\hsize]{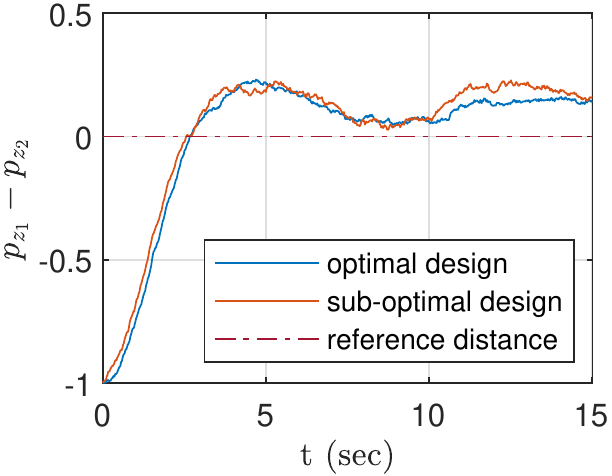}\label{perform_d}}
\caption{Control performance comparison between the optimal quantization design and a sub-optimal quantization design; (a) Change in sub-optimality; (b) Relative distance in $X$-axis; (c) Relative distance in $Y$-axis; (d) Relative distance in $Z$-axis.}
\label{global error upper bound}
\end{figure}

The sub-optimal solutions obtained with the optimal and sub-optimal quantization design are shown in In Fig. 4(a). It can be seen that the sub-optimality from the optimal quantization design has an order around $10^{-2}$ while the one from the sub-optimal quantization design has an order around $10^{-1}$. In Fig. 4(b)-(d), the relative distances between the leader and agent 2 in the $X$-axis, $Y$-axis and $Z$-axis are shown. As shown in the dark-red line, for the $Y$-axis, the desired relative distance can be achieved with the optimal quantization design. For comparison, worse tracking performance for the relative distance can be observed with the sup-optimal quantization design. At $t = 10$, the tracking error is 0.015 m when using the optimal quantization design, while the tracking error is 0.136 m when using the sub-optimal quantization design. In the $Z$-axis, the optimal quantization design outperforms the sub-optimal design after $t=10$. In the $X$-axis, both quantization designs yield similar performance, because movement in the $X$-axis is relatively slow. In conclusion, while the change between optimal solutions is upper bounded by the same $\rho$ for both quantization designs, the optimal quantization design is able to consistently achieve better sub-optimality, which leads to better control performance.

\section*{Appendix}

The coefficients $a_1, a_2, a_3$ and $b_1, b_2, b_3$ from Assumption \ref{ap5} are defined as follows:

\begin{align*}
    & a_{1}=(\kappa+1)(\kappa), a_{2}=(M \sqrt{m_{max}} \kappa(\kappa+1)\left(d L_{\max }+\sqrt{L}\right)\\
    & +M \sqrt{m_{max}} L(\kappa+\gamma-1)(1-\gamma))(L \kappa(\kappa+\gamma-1)(1-\gamma)),\\
    & a_{3}=(M \sqrt{d m_{max}}(\kappa+1))(L(\kappa+\gamma-1)(1-\gamma)), \\
   & b_{3} =(L_{\max } M \sqrt{d m_{max}} \kappa(\kappa+1)+L \sqrt{d m_{max}}(\kappa+\gamma-1)\\
   &(1-\gamma)) (L \kappa(\kappa+\gamma-1)(1-\gamma)), b_{1} =(L_{\max }(\kappa+1))(\kappa),  \\
   & b_{2}= (L_{\max } M \sqrt{m_{max}} \kappa(\kappa+1)\left(d L_{\max }+\sqrt{L}\right)+L_{\max } d \\
   &  \sqrt{m_{max}}  L(\kappa+1)(\kappa+\gamma-1)(1-\gamma))(L \kappa(\kappa+\gamma-1)(1-\gamma)).\\
\end{align*}

The Coriolis force matrix, damping matrix and gravity matrix defined in the nonlinear model \eqref{nonlinear model} are as follows:
\begin{align*}
C(\nu) &= 
\begin{bsmallmatrix}
0 & 0 & 0 & 0 & M_{z} v_z & -M_{y} v_y \\
0 & 0 & 0 & -M_{z} v_z & 0 & M_{x} v_x \\
0 & 0 & 0 & M_{y} v_y & -M_{x} v_x & 0 \\
0 & M_{z} v_z & -M_{y} v_y & 0 & M_{\omega_x} \omega_z & -M_{\omega_y} \omega_y \\
-M_{z} v_z & 0 & M_{x} v_x & -M_{\omega_x} \omega_z & 0 & M_{\omega_z} \omega_x \\
M_{y} v_y & -M_{x} v_x & 0 & M_{\omega_y} \omega_y & -M_{\omega_z} \omega_x & 0
\end{bsmallmatrix} , \allowdisplaybreaks\\ 
 D(\nu) &=  
\operatorname{diag} ([X_{v_x|v_x|}|v_x|,Y_{v_y|v_y|}|v_y|,Z_{v_z|v_z|}|v_z|, \nonumber\\
& \quad\quad\quad\quad K_{\omega_x|\omega_x|}|\omega_x|,M_{\omega_y|\omega_y|}|\omega_y|,N_{\omega_z|\omega_z|}|\omega_z|]) ,\allowdisplaybreaks\\
 g(\eta) &= 
\begin{bmatrix}
(W-B) q_y \\
-(W-B) \cos q_y  \sin q_z  \\
-(W-B) \cos q_y  \cos q_z  \\
\mathbf{0}_{3 \times 1}
\end{bmatrix},
\end{align*}
where $\mathbf{0}_{3 \times 1}$ is zero column matrix. $W$ and $B$ represent the weight and buoyant force acting on the AUV while assuming center of gravity collocates with the center of buoyancy.

\bibliographystyle{IEEEtran}
\bibliography{IEEEabrv,MC}

\end{document}